\documentclass[11pt,thmsa]{article}%
\usepackage{amsfonts}
\usepackage{amsmath}
\usepackage{amsthm}
\usepackage{latexsym}
\usepackage{amssymb}%
\usepackage{graphicx}
\usepackage{color}

\usepackage[utf8x]{inputenc}

\usepackage{tikz}
\usepackage{color}
\usepackage{graphicx}
\providecommand{\U}[1]{\protect\rule{.1in}{.1in}}

\newtheorem{theorem}{Theorem}
\newtheorem{lemma}[theorem]{Lemma}

\newtheorem{proposition}[theorem]{Proposition}
\newtheorem{corollary}[theorem]{Corollary}

\theoremstyle{definition}

\newtheorem{example}[theorem]{Example}

\begin{document}

\title{Intergenerational Preferences and Continuity: Reconciling Order and Topology}
\author{ A. Estevan{\footnotesize$^{1}$}, R.R. Maura{\footnotesize$^{2}$}, O. Valero{\footnotesize$^{3,4}$}}

\date{{\footnotesize $^{1}$Departamento EIM, Universidad P\'{u}blica de Navarra, Instituto INAMAT, Campus Arrosad\'{i}a, Pamplona, 31006, Spain. Corresponding Author. E-mail: asier.mugertza@navarra.es}\\
{\footnotesize $^{2}$Department of Economics, London School of Economics, WC2A 2AE London, United Kingdom.
E-mail: r.maura-rivero@lse.ac.uk}\\
{\footnotesize $^{3}$Departament de Ci\`encies; Matem\`atiques i Inform\`atica, Universitat de les Illes Balears, 07122 Palma de Mallorca (Illes Balears), Spain. E-mail: o.valero@uib.es}\\
{\footnotesize $^{4}$Institut d' Investigaci\'o Sanit\`aria Illes Balears (IdISBa), Hospital Universitari Son Espases, 07120 Palma de Mallorca (Illes Balears), Spain.}
}

\maketitle

\begin{abstract}
In this paper we focus our efforts on studying how a preorder and topology can be made compatible. Thus we provide a characterization of those that are continuous-compatible. Such a characterization states that such topologies must be finer than the so-called upper topology induced by the preorder and, thus, it  clarifies which topology is the smallest one among those that make the preorder continuous. Moreover, we provide sufficient conditions that allows us to discard in an easy way the continuity of a preference. In the light of the obtained results, we provide possibility counterparts of the a few celebrate impossibility theorems for continuous social social intergenerational preferences due to P. Diamond, L.G. Svensson and T. Sakai. Furthermore, we suggest quasi-pseudo-metrics as appropriate quantitative tool for reconciling topology and social intergenerational preferences. Thus, we develop a metric type method which is able to guarantee possibility counterparts of the aforesaid impossibility theorems and, in addition, it is able to give numerical quantifications of the improvement of welfare. We also show that our method makes always the intergenerational preferences semi-continuous multi-utility representables in the sense of \"{O}zg\"{u} Evern and Efe O. Ok. Finally, in order to keep close to the classical way of measuring in the literature, a refinement of the previous method is presented in such a way that metrics are involved.

\medskip\

 \medskip\
\end{abstract}

\section{Introduction\label{introduct}}
The intergenerational distribution problem have been studied in depth since the beginning of the twentieth century. In 1907, Henry Sidgwick stated that every rational distributional criterion (social intergenerational preferences) with an infinite horizon must satisfy the finite anonymity (\cite{Sidgwick1970}). Later on, in 1960, Tjalling Koopmands  added to this intergenerational equity requirement the continuity and the impatience (\cite{Koopmans1960}). Then, Peter Diamond showed that the former conditions conflict the continuity requirement in his celebrated Impossibility theorem (\cite{Diamond1965}). Concretely, the aforementioned theorem states a conflict between finite anonymity, impatience (Pareto efficiency) and the continuity with respect to the topology induced by the so-called supremum metric. The finding of Diamond caused several authors to try to discern, on the one hand, whether there exists any distributional criterion satisfying the finite anonymity and impatience at the same time and, on the other hand, whether both conditions can be compatible with continuity with respect $T_0$ any topology. In this direction, Lars-Gunnar Svensson firstly proved the existence of an intergenerational distributional criterion which fulfills simultaneously equity and Pareto efficiency (\cite{svensson}). Secondly, he explored the role of continuity and, thus, he provided an example of intergenerational distributional criterion which satisfies equity, Pareto efficiency and, in addition, continuity. However, this time the continuity was considered with respect to a topology finer than the topology induced by the supremum metric. 

Svensson did not answer completely the question about what topologies can be considered in oder to make continuous the intergenerational distributional criterion when the equity and Pareto efficiency requirements are also under consideration. Motivated, in part, by Sevensson's partial answer to the posed question, Kuntal Banerjee and Tapan Mitra addressed the problem of identifying those topologies that are compatible with equity and Pareto efficiency in \cite{mitra}. To this end, they provided a necessary condition which is expressed in terms of a simplex condition that must be satisfied by the metric inducing the topology. In this case, the considered topologies came from a collection of metrics that belong to a class whose properties are commonly used in the literature and they appear to be natural from a social decision making viewpoint. Of course, the supremum metric and the metric that induces the topology explored by Svensson belong to the aforesaid class. Besides, Banerjee and Mitra prove that, among the topologies induced by the metrics in such a class, the topology considered by Svensson is the coarsest one for which an intergenerational distributional criterion can be continuous, equity and Pareto efficiency.

In the exposed studies, the authors have considered the intergenerational equity and Pareto efficiency expressed by the so-called anonymity and strong Pareto axioms, respectively. The first requirement, anonymity, is an ethical criterion which expresses that every generation must be treated equally regardless how far they are in time. The second one, strong Pareto, exhibits sensitivity to changes in the welfare levels of each generation.  So it seems natural to wonder whether is possible to express both requirements by means of another criteria that bring compatibility with continuity.  

Regarding the strong Pareto axiom, in \cite{FM03} Marc Fleurbaey and Phillippe Michel considered the so-called weak Pareto axiom in order to express the intergenerational efficiency and showed that a stronger version of Diamond's Impossibility theorem can be deduced. Hence they proved that anonymity, weak Pareto and continuity with respect to the topology induced by the supremum metric are also incompatible.  

Toyotaka Sakai introduced a new concept of equity in \cite{sakaiSC}. Specifically, Sakai exposed that anonymity is not able to capture all aspects of intergenerational equity because this requirement expresses that present-biased and future-biased intergenerational distributions must be treated equally and it is not sensitive to balanced distributions. Motivated by this fact, he introduced the distributive fairness semiconvexity axiom which expresses that balanced distributions are preferable to the aforementioned biased intergenerational distributions. Moreover, Sakai proved again the incompatibility of anonymity, distributive fairness semiconvexity and and continuity induced by the supremum metric. Furthermore, a distributive fairness version of Svenson's possibility result was provided by Sakai when Strong Pareto requirement is replaced by (strong) distributive fairness semiconvexity. It must be stressed that the intergenerational preference constructed by Sevensson violates the strong distributive fairness semiconvexity. So the impossibility result due to Sakai is only based on intergenerational ethical requirements because no Pareto axioms are assumed. 

In \cite{sakaiEB}, Sakai introduced a new requirement, that he called sensitivity to the present, which is able to captures in some sense anonymity and distributive fairness semiconvexity in such a way that it is sensitive to changes of utility or welfare of present generations. Concretely, Sakai showed that such an axiom can be derived independently from strong Pareto requirement and from distributive fairness semiconvexity requirement and, in addition, a generalization of Diamond's and Sakai's impossibility theorems was obtained in the sense that sensitivity to the present brings an incompatibility with anonymity and continuity with respect to the supremum metric. 

Motivated by the exposed facts, in this paper we focus our efforts on studying how the intergenerational distributional criteria and the topology can be made compatible. Thus we provide a characterization of those that are continuous-complatible. Such a characterization states that such topologies must be finer than the so-called upper topology and, thus, it  clarifies which topology is the smallest one among those that make the social social intergenerational preferences continuous. Moreover, we provide sufficient conditions that allows us to discard in an easy way the continuity of a preference.

We use our characterization in order to provide possibility counterparts of the above mentioned impossibility theorems due to Diamond, Svensson and Sakai. 
Our methodology is in accordance with the classification, due to Banerjee and Mitra, of the metrics belonging to the class considered in \cite{mitra}. However, the new characterization presents two advantages with respect to the approach given in the aforesaid reference. On the one hand, the new result allows us to decide the continuity of the preference even if the topology under consideration is not metrizable. On the other hand, Banerjee and Mitra only provide a necessary condition. Hence one can find preferences that enjoy anonymity and strong Pareto requirements and, in addition, they fulfill the simplex condition in \cite{mitra} but they are not continuous. An example of this type of preferences is provided. 

The fact that the upper topology is not metrizable (notice that it is not Hausdorff) suggests us that the appropriate quantitative tool for reconciling topology and social intergenerational preferences is exactly provided by quasi-pseudo-metrics which are able to encode the order relation that induce the intergenerational preference. Observe that quasi-pseudo-metrics have  already been applied to economical problems in the literature (see, for instance, \cite{FaugerasRuschendorf18,Levin1984,Levin1997,Levin2008,Levin2011,StoyanovRachevFabozzi12,StoyanovRachevFabozzi11}). This generalized metric notion helps us to provide both things, the numerical quantifications about the increase of welfare and the arrow of such an increase. Note that a metric would be able to yield information on the increase but it, however, will not give the aforementioned arrow.

Based on the fact that every preorder, and thus every social intergenerational preference, can be encoded by means of a quasi-pseudo-metric (see, for instance, \cite{Goubault-Larrecq13}) we develop a method to induce a quasi-pseudo-metric that makes always the preference continuous with respect its induced topology,  the Alexandroff topology generated by the preorder, which is finer than the upper topology. Thus such a method is again able to guarantee possibility counterparts of the celebrate impossibility theorems due to Diamond, Svensson and Sakai and, in addition, it is able to give numerical quantifications of the improvement of welfare.

Since in economics analysis it is convenient to represent preferences through real valued functions (\cite{Mas-ColellWhinstonGreen1995,Varian1992}), the so-called utility functions, we also show that our method makes always the preferences semi-continuous multi-utility representables in the sense of \"{O}zg\"{u} Evern and Efe O. Ok  (\cite{EvO11}). 

Finally, in order to keep close to the classical way of measuring in the literature, a refinement of the previous method is presented in such a way that metrics are involved. The economical interpretations of their quantifications are also exposed.

\section{Preliminaries on preorders and intergenerational preferences}\label{prelim}

In this section we recall the basics on order theory and intergenerational preferences in decision making theory that will be useful in our subsequent discussion.

According to \cite{bridges}, a \emph{preorder} on a non-empty set $X$ is a reflexive and transitive binary relation $\precsim$ on $X$.  A preorder is called a \emph{preference} in \cite{Ok07}. A complete preorder is a \emph{rational preference} in \cite{Mas-ColellWhinstonGreen1995}. Complete preorders are also known as \emph{total preorders} in \cite{bridges,GMetha1998}. Although the preorders has usually been assumed as rational preferences in the literature, the notion of preorder has turned to be very useful in many fields of economics (for a deeper treatment of the topic we refer the reader to \cite{EvO11} and references therein).
 
From now on, given a non-empty set $X$ endowed with a preference  $\precsim$, as usual, we will denote by $x\sim y$ the fact that ($x\precsim y $ and $y\precsim x$). Moreover, $x\prec y$ will denote the fact that $x\precsim y$ and, in addition, not $y\precsim x$. When $x$ and $y$ are incomparable we will write $x\bowtie y$. Thus, $x\bowtie y$ if and only if $\neg(x\precsim y )$ as well as  $\neg(y\precsim x)$.

Following \cite{Ok07} (compare \cite{bridges} and \cite{GMetha1998}), for any $y \in X$ and any relation $\precsim$ on $X$, the \emph{contour sets} (\emph{lower} and \emph{upper}) are defined as follows: 
\begin{enumerate}
\item[(1)]$L^{\precsim}(y)=\{x\in X \colon x\precsim y\}$ (lower counter set),
\item[(2)]$U^{\precsim}(y)=\{x\in X \colon y\precsim x\}$ (upper counter set).
\end{enumerate} \medskip

On account of \cite{Goubault-Larrecq13},  a subset $G$ of a non-empty set $X$ is said to be an \emph{up-set} (or \emph{upward closed}) with respect a preorder $\precsim$ on $X$ provided that $y\in G$ whenever $x,y\in X$ with $x\in G$ and $ x\precsim  y$. Dually, a subset $G$ is said to be a \emph{down-set} (or \emph{downward closed}) with respect a preorder $\precsim$ on $X$ provided that $y\in G$ whenever $x,y\in X$ with  $x\in G $ and $ y\precsim  x$.

According to \cite{GMetha1998} (see also \cite{Ok07} and compare \cite{Mas-ColellWhinstonGreen1995,Varian1992}), a rational preference $\precsim$ on $X$ is called \emph{representable} \rm if there is a real-valued function $u\colon X\to \mathbb R$ that is order-preserving, so that, for every $x, y \in X$, it holds that 
$$x \precsim y \iff u(x) \leq u(y).$$ The map $u$ is said to be a \emph{utility function} \rm for $\precsim$. 

According to \cite{bridges} (see also \cite{GMetha1998}), a rational preference $\precsim$ on $X$ is said to be \emph{separable} (separable in the sense of Debreu in \cite{bridges}) if there exists a countable subset $D \subseteq X$ such that for every $x,y \in X$ with $x \prec y$ there exists $d \in D$ such that $x \precsim d \precsim y$. In the case of separable rational preferences we have it is representable if and only if it is separable.

When the preorder is not total, then a representation can also be proposed. Hence, according to \cite{Ok07,Peleg1970,Ritcher1966}, a preorder is \emph{Ritger-Peleg representable} if there is a function $u\colon X\to \mathbb R$ that is strictly isotonic, so that, for every $x, y \in X$, it holds that 

$$x \precsim y \implies u(x) \leq u(y) \mbox{ and } x\prec y \implies u(x)<u(y).$$ The map $u$ is said to be a \emph{Ritcher-Peleg utility function} for $\precsim$. 

Obviously, a Ritcher-Peleg representation does not characterize the preorder, i.e., the preorder cannot be retrieved, in general, from the Ritcher-Peleg utility function. Motivated by this fact, the multi-utility representation was introduced in \cite{EvO11} (see, also \cite{Levin, Levin2}). In particular, a preorder $\precsim$ on a set $X$ is said to have a {\em multi-utility representation} if  there exists a family $\mathcal{U}$ of isotonic real-valued functions (\emph{weak-utilities}) such that 
for all points $x ,y \in X$  the following equivalence holds:

\begin{equation} \label{mult1} x \precsim y \Leftrightarrow  \forall u \in {\cal U} \,\,(u(x) \leq u(y))\end{equation} Observe that the members of a  multi-utility representation $\mathcal{U}$ are isotonic but they do not need to be strict isotonic in general. This fact make different the multi-utility representation from Ritger-Peleg utility representation. It must be pointed out that a rational preference admits a multi-utility representation even when it is not separable and, thus, it does not admit a utility representation.

The advantage of the multi-utility representation with respect to the above exposed type of representations is twofold. On the one hand, it always exists (see Proposition 1 in \cite{EvO11}). On the other hand, it fully characterizes the preorder.

When discussing about intergenerational distribution criteria the following axioms can be assumed to be satisfied for those preorders that are applied to rank the different alternativies. In the literature a few alternative sets are considered and, usually, all of them are subsets of the set $l_{\infty}=\{(x_n)_{n\in\mathbb{N}}: x_i \text{ with } \sup_{i\in\mathbb{N}}x_i<\infty\}$. 

Let us recall that the most usual alternative sets are $$l_{\infty}^+=\{(x_n)_{n\in\mathbb{N}}\in l_{\infty}: x_i\geq 0 \text{ for all } i\in \mathbb{N}\}$$ and $$l_{\infty}^{[0,1]}=\{(x_n)_{n\in\mathbb{N}}\in l_{\infty}: 0\leq x_i\leq 1 \text{ for all } i\in \mathbb{N}\}.$$ Let us recall that the alternative sets $l_{\infty}^+$ and $l_{\infty}^{[0,1]}$ have been considered, for instance, in \cite{Campbell1985,Epstein1986,FM03,sakaiSC} and \cite{mitra,Diamond1965,sakaiEB,svensson}, respectively. However, the whole space $l_{\infty}$ has been considered in \cite{mitra2,Lawers1997, Sakai06}. 

From now on, an alternative set will be any subset $X$ of $l_{\infty}$, i.e., $X\subseteq l_{\infty}$. Next we recall the below concepts which will play a crucial role in order to state possibility theorems later on. We refer the reader, for instance, to  \cite{mitra,sakaiSC}.

A \emph{finite permutation} is a bijection $\pi\colon \mathbb{N} \to \mathbb{N}$ such that there is  $t_0 \in\mathbb{N}$ satisfying $t = \pi(t),\forall t>t_0$. In the sequel, $\Pi_{\infty}$ will denote the set of all such $\pi$.

A preorder $\precsim$ on $X$ is said to satisfy the \emph{anonymity axiom} if and only if $x\sim\pi(x)$ for all $x \in X$ and for all $\pi\in\Pi$. Anonymity expresses that every generation must be treated equally regardless how far they are in time. However, as exposed in Introduction, such an axiom does not capture all aspects of intergenerational equity because it is not sensitive to balanced distributions. In order to avoid this handicap, \emph{distributive fairness semiconvexity axiom} has been considered. This axiom  expresses that balanced distributions are preferable to the aforementioned biased intergenerational distributions and it can be states as follows:

A preorder $\precsim$ on $X$ is said to satisfy the distributive fairness semiconvexity axiom if and only if for all $x \in X$ and for all $\pi\in\Pi$ we have that  there exists $s\in (0, 1)$ such that $sx + (1-s)\pi(x) \succ x, \pi(x)$ whenever $x\neq\pi(x)$. Moreover, a stronger version of the previous axioms can be expressed via the \emph{strong distributive fairness semiconvexity} which states that a preorder $\precsim$ on $X$ satisfies the strong distributive fairness semiconvexity axiom if and only if for all $x \in X$ and for all $\pi\in\Pi$ we have that  $sx + (1-s)\pi(x) \succ x, \pi(x)$ for all $s\in (0, 1)$ whenever $x\neq\pi(x)$.

An axiom which captures sensitivity to changes in the welfare levels of each generation is called \emph{weak monotonicity axiom} or \emph{weak Pareto axiom}. It can be stated in the following way: 

A preorder $\precsim$ on $X$ is said to be  \emph{weak monotone} or \emph{weak Pareto} if and only if, for all $x,y\in X$, $x \prec y$ provided that $x_t< y_t$ for all $t\in \mathbb{N}$. A stronger version of weak monotonicity axiom is the \emph{strong monotonicity axiom} or \emph{strong Pareto axiom}. Thus, a preorder $\precsim$ is said to be  \emph{strong monotone} or \emph{strong Pareto} if and only if, for all $x,y\in X$, $x \prec y$ provided that $x_t\leq y_t$ for all $t\in \mathbb{N}$ and, in addition, $x \not = y$. Clearly, every strong Pareto preorder is always weak Pareto.

Sensitivity to the present is an axiom which is able to capture, in some sense, anonymity and distributive fairness semiconvexity in such a way that the preorder is sensitive to changes of utility or welfare of present generations. Formally, a preorder $\precsim$ on $X$ satisfies sensitivity to the present provided that, for each $x\in X$, there are $y,z\in X$ and $t\in \mathbb{N}$ such that $(z^{t},^{t+1}x) \prec (y^{t},^{t+1}x)$, where, for each $w\in X$, $(w^{t},^{t+1}x)_{i}=w_i$ for all $i\in\mathbb{N}$ with $i\leq t$ and, in addition, $(w^{t},^{t+1}x)_{i}=x_i$ for all $i\in\mathbb{N}$ with $t+1\leq i$.

In the remainder of the paper, a preorder on $X$ fulfilling any equity requirement (anonymity, distributive fairness semiconvexity or sensitivity to the present) and any monotony (strong or weak) will be called an ethical social welfare preorder. An ethical social welfare preorder that is a rational preference (complete preorder) will be called ethical social welfare order (ethical preference in \cite{svensson}). It is worthy to mention that ethical social welfare preorders and ethical social welfare orders have ben shown to exist in \cite{sakaiSC,svensson}.

\section{The continuity of preferences: a characterization and possibilities theorems}
In this section we study the way through which the intergenerational preferences and the topology can be made compatible. Since two notions of continuity have been taken into account in the intergenerational distribution problems. We provide a characterization of both type of continuities and they are independent of any equity or Pareto requirement. Moreover, we  clarify which topology is the smallest one among those that make the  preorder continuous in both senses. This allows us to solve an open problem in the literature. Partial answers to such a problem have been given by means of the so-called impossibilities theorems which state that there does not exist any ethical social welfare (pre)order which is continuous with respect the topology under consideration (mainly the product topology or the supremum topology on $X\subseteq l_\infty$). Accomplished this item, we apply our characterization in order to get possibility counterparts of the aforementioned impossibility theorems due to Diamond, Svensson and Sakai.

\subsection{The characterization}\label{NewCharact}
First we recall a few pertinent notions from topology that will be very useful in order to achieve our target. 

According to \cite{Goubault-Larrecq13} (see also \cite{Gierzetall2003}), a preorder can be always induced on a topological space $(Y,\tau)$. Such a preorder $\precsim_{\tau}$ is called the \emph{specialization preorder} induced by $\tau$ and it is defined as follows: 

$$x\precsim_{\tau}y \Leftrightarrow \text{ every open subset containing } x \text{ also contains }y.$$ It is not hard to check that $x\precsim_{\tau}y \Leftrightarrow x\in cl_{\tau}(\{y\})$, where by  $cl_{\tau}(\{y\})$  we denote the closure of $\{y\}$ with respect to $\tau$.

It is clear that the specialization preorder allows us to achieve a preorder from every topology. It is known too that every preorder can be obtained as a specialization preorder of some topology (\cite{Goubault-Larrecq13}). However, the correspondence is not bijective, since there are in general many topologies on a set $X$ which induce a given preorder $\precsim$ as their specialization preorder. Among the aforesaid topologies we find the upper topology and the Alexandroff topology. The first one is the coarsest  topology and the second one is the finest topology that induce the preorder $\precsim$ as their specialization preorder. Notice that there are many other topologies between them and that, in general, the Alexandroff and the upper topologies does not coincide. An example that shows that the upper topology and the Alexandroff topology are not the same in general can be found in \cite[Example 1]{BEstRP19}.

Let us recall that, given a preorder $\precsim$ on a non-empty set $X$, the \emph{upper topology} $\tau^{\precsim}_{U}$ is defined as that which has the lower contour set $L^{\precsim}(x)$ closed ($x\in X$), that is, $\tau^{\precsim}_{U}$  is the topology arised from the subbase  $\{Y\setminus L^{\precsim}(x) \}_{x\in X}$. Observe that a preorder $\precsim^{-1}$ can be induced from a preorder $\precsim$ on $X$  as follows: $x \precsim^{-1} y \Leftrightarrow y\precsim x$. The preorder $\precsim^{-1}$ is called the dual preorder or the opposite of $\precsim$. Clearly $L^{\precsim^{-1}}(y)=U^{\precsim}(y)$ for all $y\in Y$. Taking this into account, we will denote by $\tau^{\precsim}_{L}$ the upper topology on $Y$ induced by $ \precsim^{-1}$. Notice that such a topology matches up with the \emph{lower topology} induced by $\precsim$ on $X$, that is, the topology whose subbase is $\{Y\setminus U^{\precsim}(y)\}_{y\in X}$. 


Usually intergenerational preferences are assumed to satisfy that two intertemporal distribution that are not very different must be have similar welfare levels. This is  accomplished by assuming that the preorder under consideration is continuous. Let us recall the two usual notions of continuity.

A preorder $\precsim$ on a topological space $(Y,\tau)$ is said to be \emph{$\tau$-continuous} if, for all $y\in Y$, the lower contour $L^{\precsim}(x)$ and the upper contour $U^{\precsim}(x)$ are closed with respecto to $\tau$ (see, for instance, \cite{Diamond1965,sakaiSC,sakaiEB,svensson}). However, a weak form of continuity is sated in the literature, the so-called lower continuity (among others, see \cite{mitra,EvO11,Sakai06}). Thus, a preorder on a topological space is said to be \emph{lower $\tau$-continuous} 
provided that, for all $y\in Y$, the lower contour $L^{\precsim}(x)$ is closed with respecto to $\tau$. 
 \medskip

From now on, given a preorder $\precsim$ on $Y$ and $x_1,\ldots,x_n\in Y$, we will set $$\downarrow_{\precsim} \{x_1,\ldots,x_n\}=\{z \text{ such that there exists } i\in \{1,\ldots,n\} \text{ with } z\precsim x_i\}.$$ Dually $\uparrow_{\precsim} \{x_1,\ldots,x_n\}$ can be defined. 
 \medskip
 
In view of the exposed facts we introduced the promised characterization of both type of continuities.  

\begin{theorem}\label{Cont1}Let $\precsim$ be a preorder on a topological space $(Y,\tau)$. Then the following assertions are equivalent:

\begin{enumerate}

\item[(1)] $\precsim$ is $\tau$-continuous.

\item[(2)] The topology $\tau$ is finer than the coarsest topology including $\tau^{\precsim}_{U}$ and $\tau^{\precsim}_{L}$.

\end{enumerate}
\end{theorem}

\begin{proof}$(1)\Rightarrow (2)$. First we show that $\tau^{\precsim}_{L}\subseteq \tau$. Let $A\in  \tau^{\precsim}_U$. Then, given $x\in A$, there exist $x_1,\ldots,x_n\in X$ such that $x\in Y\setminus \downarrow_{\precsim} \{x_1,\ldots,x_n\}\subseteq A$. Moreover, $Y\setminus \downarrow_{\precsim} \{x_1,\ldots,x_n\}=Y\setminus \bigcup_{i=1}^{n}L^{\precsim}(x_i)=\bigcap_{i=1}^{n}Y\setminus L^{\precsim}(x_i)$. Since $Y\setminus L^{\precsim}(x_i)\in \tau$ we deduce that $\bigcap_{i=1}^{n}X\setminus L^{\precsim}(x_i)\in \tau$. Then $x\in \bigcap_{i=1}^{n}X\setminus L^{\precsim}(x_i)\subseteq A$. It follows that $A\in \tau$. Hence $\tau^{\precsim}_{U}\subseteq \tau$.
\medskip

Next we show that $\tau^{\precsim}_{L}\subseteq \tau$. To this end,  let $A\in  \tau^{\precsim}_L$. Then, given $x\in A$, there exist $x_1,\ldots,x_n\in X$ such that $x\in Y\setminus \uparrow_{\precsim} \{x_1,\ldots,x_n\}\subseteq A$. Moreover, $Y\setminus \uparrow_{\precsim} \{x_1,\ldots,x_n\}=Y\setminus \bigcup_{i=1}^{n}U^{\precsim}(x_i)=\bigcap_{i=1}^{n}X\setminus U^{\precsim}(x_i)$. Since $Y\setminus U^{\precsim}(x_i)\in \tau$ we deduce that $\bigcap_{i=1}^{n}X\setminus U^{\precsim}(x_i)\in \tau$. Then $x\in \bigcap_{i=1}^{n}X\setminus U^{\precsim}(x_i)\subseteq A$. It follows that $A\in \tau$. Hence $\tau^{\precsim}_{L}\subseteq \tau$.
\medskip

The precedig facts joint with the fact that the coarsest topology including $\tau^{\precsim}_{U}$ and $\tau^{\precsim}_{L}$ is formed by all finite intersection of elements in $\{\tau^{\precsim}_{U},\tau^{\precsim}_{L}\}$ and all arbitrary unions of these finite intersections gives immediately that $\tau$ is a topology finer than it. 
\medskip

$(2)\Rightarrow (1)$. Take $y\in X$ and consider $L^{\precsim}(y)$ and $U^{\precsim}(y)$. Then $Y\setminus L^{\precsim}(y)\in \tau^{\precsim}_U$ and $Y\setminus U^{\precsim}(y)\in \tau^{\precsim}_L$. Since $\tau^{\precsim}_{U},\tau^{\precsim}_{L}\subseteq  \tau$ we deduce that $Y\setminus L^{\precsim}(y),Y\setminus U^{\precsim}(y)\in \tau$. It follows that $\precsim$ is $\tau$-continuous.
\end{proof}
\medskip

The next result, which characterizes the lower continuity, can be found in \cite[Corollary 1]{BEstRP19}. Although it was stated without proof in the aforesaid reference, we have omitted its proof because it follows similar arguments to those given in the proof of Theorem \ref{Cont1}.

\begin{theorem}\label{Cont2}Let $\precsim$ be a preorder on a topological space $(Y,\tau)$. Then the following assertions are equivalent:

\begin{enumerate}

\item[(1)] $\precsim$ is lower $\tau$-continuous.
\item[(2)]  The topology $\tau$ is finer than $\tau^{\precsim}_{U}$.

\end{enumerate}
\end{theorem}

The preceding characterizations state that the topologies that can be taken under consideration in order to make, on the one hand,  continuous the preorder must be finer than the coarsest topology including the upper and lower topologies induced by the preorder and, on the other hand, must be finer than the upper topology with the aim of warranting the lower continuity.  Thus, it clarifies which topology is the smallest one among those that guarantee such continuities. In the light of this, it makes no sense to work with a topologies which does not refine the aforesaid ones. 

Notice that these results turn out key when the continuity of ethical social welfare preorders and orders is discussed. This fact will be exploited in the the next subsection where we introduce possibility theorems, i.e., theorems that reconcile  social welfare (pre)orders and the topology making them continuous. Observe that the preceding results, on the one hand, answer to a question that has been discussed a lot in the literature and, on the other hand, improves the result given in \cite[Theorem 1]{mitra}.
\medskip

Going back to the specialization preorder, let us recall that, given a preorder $\precsim$ on $Y$, the {\em Alexandroff topology} $\tau^{\precsim}_A$ is formed by all {up-sets} with respect to $\precsim$.  Observe that the lower sets are   closed sets with respect to $\tau^{\precsim}_A$.

From the preceding characterizations we obtain the following ones which give sufficient conditions to make continuous a preorder.

\begin{corollary}\label{Alex1} Let $\precsim$ be a preorder on a topological space $(Y,\tau)$. If $\tau$ is finer than $\tau^{\precsim}_A$ and $\tau^{\precsim^{-1}}_A$, then $\precsim$ is $\tau$-continuous.
\end{corollary}

\begin{proof} Since $\tau^{\precsim}_U\subseteq \tau^{\precsim}_A$ and $\tau^{\precsim}_L\subseteq \tau^{\precsim^{-1}}_A$ we conclude, from Theorem \ref{Cont1}, that $\precsim$ is $\tau$-continuous.

\end{proof}

\begin{corollary}\label{Alex2} Let $\precsim$ be a preorder on a topological space $(Y,\tau)$. If $\tau$ is finer than $\tau^{\precsim}_A$, then $\precsim$ is lower $\tau$-continuous.
\end{corollary}

\begin{proof} Since $\tau^{\precsim}_U\subseteq \tau^{\precsim}_A$ we conclude, from Theorem \ref{Cont2}, that $\precsim$ is lower $\tau$-continuous.

\end{proof}

The next example shows that the converse of Corollaries \ref{Alex1} and \ref{Alex2} do not hold in general. In order to introduce such an example, notice that a sequence $(x_n)_{n\in\mathbb{N}}$ in $Y$ converges to $x\in Y$ with respect to $\tau^{\precsim}_A$ if and only if there exists $n_0\in\mathbb{N}$ such that $x\precsim x_n$ for all $n\geq n_0$.

\begin{example}\label{SvenEx1}Consider the preorder $\precsim$ on $l_{\infty}^{[0,1]}$ defined by 
$$y\precsim x \Leftrightarrow y_t\leq x_t \text{ for all } t\in \mathbb{N}.$$ Then $\precsim$ is $\tau_{d_s}$-continuos and, thus, lower $\tau_{d_s}$-continuos, where $d_s$ stands for the restriction of the supremum metric on $l_{\infty}$ to $l_{\infty}^{[0,1]}$, i.e., $d_s(x,y)=\sup_{t\in\mathbb{N}} |x_t-y_t|$ for all $x,y\in l_{\infty}$.

Next we show that $\tau^{\precsim}_A\not\subseteq \tau_{d_s}$. Indeed,  set $x=(0,1,0,\frac{1}{2}, \frac{2}{2}, 0, \frac{1}{3}, \frac{2}{3}, \frac{3}{3}, 0,...)$ and $l=(1,1,0,\frac{1}{2}, \frac{2}{2}, 0, \frac{1}{3}, \frac{2}{3}, \frac{3}{3}, 0,...)$. Now the sequence $(y_n)_{n\in\mathbb{N}}$ is defined as follows:
 \begin{center}
 $ y_1=x=(0,1,0,\frac{1}{2}, \frac{2}{2}, 0, \frac{1}{3}, \frac{2}{3}, \frac{3}{3}, 0,...)$,\\
 $ y_2=(\frac{2}{2},1,0,0, \frac{1}{2}, 0, \frac{1}{3}, \frac{2}{3}, \frac{3}{3}, 0,...)$,\\
 $ y_3=(\frac{3}{3},1,0,\frac{1}{2}, \frac{2}{2}, 0, 0, \frac{1}{3}, \frac{2}{3}, 0,...)$,\\
 $...$\\
 $y_n=(\frac{n}{n},1,0, \frac{1}{2}, \frac{2}{2}, 0, \frac{1}{3}, \frac{2}{3}, \frac{3}{3}, 0,...,0,0, \frac{1}{n},..., \frac{n-1}{n},0,...)$
 \end{center}

Clearly the sequence $(y_n)_{n\in \mathbb{N}}$ converges to $l=(1,1,0,\frac{1}{2}, \frac{2}{2}, 0, \frac{1}{3}, \frac{2}{3}, \frac{3}{3}, 0,...)$ on $\tau_s$, since $d_s(l,y_n)=\frac{1}{n}$. However the sequence fails to converge in $\tau^{\precsim}_A$, since $l\not\precsim y_t$ for any $t\in \mathbb{N}$.
 \end{example}

As exposed before, in economics analysis it is convenient to represent preorders through real valued functions (\cite{Mas-ColellWhinstonGreen1995,Varian1992}). We end this subsection giving conditions so that a preorder admits a semi-continuous multi-utility representation in the sense of \"{O}zg\"{u} Evern and Efe O. Ok  (\cite{EvO11}).

Let us recall that, given a topological space $(Y,\tau)$, a function $f\colon Y\to \mathbb{R}$ which is continuous from $(Y,\tau)$ into $(\mathbb{R},\tau^{\leq}_U)$ is said to be 
\emph{lower semi-continuous}. 

According to \cite[Proposition 2]{EvO11}, every (pre)order $\precsim$ on a topological space $(Y,\tau)$ which is lower $\tau$-continuous always have a multi-utility representation $\mathcal{U}$ of isotonic real-valued functions such that every member belonging to $\mathcal{U}$ is a lower semi-continuous function.  

In the light of Theorems \ref{Cont1} and \ref{Cont2} we conclude stating that every (pre)order $\precsim$ on a topological space $(Y,\tau)$ admits a semi-continuous multi-utility representation provided that $\tau$ is finer than  $\tau^{\precsim}_{U}$.

\subsection{The possibilities theorems}\label{Possi}
Based on our characterizations we present possibility counterparts of the impossibility theorems due to Diamond, Svensson and Sakai. Besides we show that our results are in accordance with the classification, due to Banerjee and Mitra, of the metrics belonging to the class considered in \cite{mitra}.  Nevertheless, we will show that our new characterizations presents two advantages with respect to the approach given in the aforesaid reference. 

As exposed before, Diamond showed in his celebrated Impossibility theorem a conflict between the fact that a preorder satisfies the finite anonymity, strong monotonicity and the continuity with respect to the topology induced by the supremum metric $\tau_{d_s}$, with $d_s(x,y)=\sup_{i\in\mathbb{N}}|x_t-y_t|$ for all $x,y\in l_{\infty}$,  (\cite{Diamond1965}). The aforesaid theorem can be stated as follows.

\begin{theorem}\label{Diamond}There is no any ethical social welfare (pre)order $\precsim$ on $l_{\infty}^{[0,1]}$ which satisfies anonymity, strong monotonicity and $\tau_{d_{s}}$-continuity.
\end{theorem}

The Diamond's impossibility theorem was extended to the case of preorders fulfilling weak monotonicity by Fleurbaey and Michel in  \cite{FM03}. Concretely they proved the next result. 
  
\begin{theorem}\label{Diamond2}There is no any ethical social welfare (pre)order $\precsim$ on $l_{\infty}^{+}$ which satisfies anonymity, weak monotonicity and $\tau_{d_{s}}$-continuity.
\end{theorem}

In \cite{sakaiSC}, Sakai introduced introduced the distributive fairness semiconvexity in order to overcome the lack of sensitivity of anonymity to balanced distributions. He proved again the incompatibility of anonymity, distributive fairness semiconvexity and continuity induced by the supremum metric. Specifically the next result was obtained.

\begin{theorem}\label{Sakai}There is no any ethical social welfare (pre)order  $\precsim$ on $l_{\infty}^{+}$ which satisfies anonymity, distributive fairness semiconvexity and $\tau_{d_{s}}$-continuity. 
\end{theorem}

It is obvious that if there is no preorder on $l_{\infty}^{+}$ satisfying anonymity, distributive fairness semiconvexity and $\tau_{d_{s}}$-continuity, then there is no any preorder fulfilling  anonymity, strong distributive fairness semiconvexity and $\tau_{d_{s}}$-continuity.
 
\medskip

The next impossibility result can be also obtained.
 
 \begin{theorem}
 There is no any ethical social welfare (pre)order $\precsim$ on $l_{\infty}^{+}$ satisfying anonymity, distributive fairness semiconvexity and lower $\tau_{d_1}$-continuity with $d_1(x,y)=\min\{1,\sum_{i=1}^{\infty} |x_t-y_t|\}$ for all $x,y\in l_{\infty}$.
 \end{theorem}
 \begin{proof}
 The same argument to those given in \cite[Lemma 1]{sakaiSC} apply here, but now, defining $G(n)$ by the finite sequence $\frac{1}{2^n}, \frac{2}{2^n}, ...,\frac{n}{2^n}$. It is enough to choose $n\in \mathbb{N}$ such that $\frac{n}{2^n}<\min\{\epsilon, s\}$ and the integer $m(n)$ satisfying that $ m(n)/2^n\leq s<  (m(n)+1)/2^n$.  \end{proof}

\medskip

Later on Sakai introduced the sensitivity to the present axiom in order to capture in some sense anonymity and distributive fairness semiconvexity at the same time. Again an incompatibility was showed in such a way that the following impossibility result, which generalizes Diamond's and Sakai's impossibility theorems, was proved. 

\begin{theorem}\label{Sakai2}There is no any ethical social welfare (pre)order  $\precsim$ on $l_{\infty}^{[0,1]}$ which satisfies anonymity, sensitivity to the present and $\tau_{d_{s}}$-continuity. 
\end{theorem}

Of course from the preceding results in which the sets $l_{\infty}^{+}$ and $l_{\infty}^{[0,1]}$ have been fixed the alternative set, one can infer the same impossibilities results considering $l_{\infty}$ as the alternative set.

Keeping in mind the characterizations disclosed in Subsection \ref{NewCharact} the following possibilities results can be stated. Recall that $X$ is any subset of $l_{\infty}$.

\begin{theorem}\label{DiamondNew}Let $\tau$ be any topology on $X$. Then the following assertions hold: 

\begin{enumerate}
\item[(1)] There exists an ethical social welfare order  $\precsim$ on $X$ which satisfies anonymity, strong monotonicity and $\tau$-continuity provided that $\tau$ is finer than the coarsest topology including $\tau^{\precsim}_{U}$ and $\tau^{\precsim}_{L}$. 

\item[(2)] There exists an ethical social welfare order $\precsim$ on $X$ which satisfies anonymity, strong monotonicity and {lower} $\tau$-continuity provided that $\tau$ is finer than $\tau^{\precsim}_{U}$.
\end{enumerate}
\end{theorem}

\begin{proof}$(1)$. According to \cite{Sakai06}, the following type of overtaking criterion $\precsim$ can be extended to an ethical social welfare order on $l_{\infty}$ and, thus, on $X\subseteq l_{\infty}$ which satisfies anonymity and strong monotonicity: 

$$y\precsim x \Leftrightarrow \text{ there is } t_0\in \mathbb{N} \text{ such that } \sum_{i=1}^{t} \left (g(x_t)-g(y_t) \right)\geq 0 \text{ for all } t\geq t_0,$$ where $g:\mathbb{R}\rightarrow \mathbb{R}^+$ is any strictly concave and strictly isotonic function.  Theorem \ref{Cont1} gives the $\tau$-continuity of such an extension. 
\medskip

$(2)$. Theorem \ref{Cont2} gives the lower $\tau$-continuity of the ethical social welfare order given in the proof of the assertion $(1)$.
\end{proof}
\medskip

Since strong monotonicity implies weak monotonicity, Theorem \ref{DiamondNew} gives as a consequence the existence of ethical social welfare order satisfying anonymity, weak monotonicity and (lower) $\tau$-continuity provided that $\tau$ is finer than the coarsest topology including $\tau^{\precsim}_{U}$ and $\tau^{\precsim}_{L}$ ($\tau^{\precsim}_{U}\subseteq \tau$).

\begin{theorem}\label{SakaiNew}Let $\tau$ be any topology on $X$. Then the following assertions hold: 

\begin{enumerate}
\item[(1)] There exists an ethical social welfare order  $\precsim$ on $X$ which satisfies anonymity, distributive fairness semiconvexity and $\tau$-continuity provided that $\tau$ is finer than the coarsest topology including $\tau^{\precsim}_{U}$ and $\tau^{\precsim}_{L}$. 

\item[(2)] There exists an ethical social welfare order $\precsim$ on $X$ which satisfies anonymity, distributive fairness semiconvexity and {lower} $\tau$-continuity provided that $\tau$ is finer than $\tau^{\precsim}_{U}$.
\end{enumerate}
\end{theorem}

\begin{proof}$(1)$. Consider again the ethical social welfare order introduced in the proof of Theorem \ref{DiamondNew} which satisfies anonymity.  In \cite[Theorem 1]{sakaiSC}, it has been proved that such an ethical social welfare order fulfills distributive fairness semiconvexity on $l_{\infty}^{+}$. Following the same arguments to those given in the aforesaid reference one can proved the distributive fairness semiconvexity on $l_{\infty}$ and, hence, on any $X\subseteq l_{\infty}$. Again by Theorem \ref{Cont1} we have the $\tau$-continuity. 
\medskip

$(2)$. Theorem \ref{Cont2} gives the lower $\tau$-continuity of the ethical social welfare order given in the proof of the assertion $(1)$.

\end{proof}

\begin{theorem}\label{SakaiNew2}Let $\tau$ be any topology on $X$. Then the following assertions hold: 

\begin{enumerate}
\item[(1)] There exists an ethical social welfare order  $\precsim$ on $X$ which satisfies anonymity, sensitivity to the present and $\tau$-continuity provided that $\tau$ is finer than the coarsest topology including $\tau^{\precsim}_{U}$ and $\tau^{\precsim}_{L}$. 

\item[(2)] There exists an ethical social welfare order $\precsim$ on $X$ which satisfies anonymity, sensitivity to the present and {lower} $\tau$-continuity provided that $\tau$ is finer than $\tau^{\precsim}_{U}$.
\end{enumerate}
\end{theorem}

\begin{proof}Since strong monotonicity implies sensitivity to the present of any preorder defined on $l_{\infty}$, the ethical social welfare order provided in the proof of Theorem \ref{DiamondNew} satisfies all the requirements demanded in assertion $(1)$ and $(2)$.
\end{proof}
\medskip

Banerjee and Mitra addressed the problem of identifying those topologies that make an ethical social welfare order continuous when anonymity and strong monotonicity are assumed (\cite{mitra}). They provided a necessary condition which is expressed in terms of a simplex condition that must be satisfied by the metric inducing the topology.
To this end, they consider a class $\triangle$ of metrics which satisfy four properties that we will not expose here because they are not relevant for our purpose. For a depper discussion of such properties we refer the reader to \cite{mitra}. 

Although the considered class imposes constraints about the metrics, the most usual metrics applied to the intergenerational distribution problem belong to  $\triangle$. Concretely the following celebrated metrics on $l_{\infty}$ are in the aforementioned class: $d_c,d_s,d_p,d_1,d_q$, where

$$\begin{array}{ll}
d_c(x,y)=\sum_{i=1}^{\infty} \frac{|x_{t}-y_{t}|}{2^{i}} &.\\
\\
d_p(x,y)=\min\{1,(\sum_{i=1}^{\infty} \left |x_t-y_t|^{p}\right )^{\frac{1}{p}}\}  \text{ with } p\in ]1,\infty[.& \\
\\
d_q(x,y)=\min\{1,\sum_{i=1}^{\infty} \left (|x_t-y_t|^{q}\right )\} \text{ with } q\in ]0,1[.& \\
\end{array}
$$

Notice that $\tau_{d_{c}}\subseteq \tau_{d_{s}}\subseteq \tau_{d_{p}}\subseteq \tau_{d_{1}} \subseteq \tau_{d_{q}}$. 
\medskip

The Banerjee and Mitra result can be stated as follows:

\begin{proposition}\label{Simplex}Let $d\in \triangle$ and let $\precsim$ be an ethical social welfare preorder on $l_{\infty}^{[0,1]}$ which satisfies anonymity and strong monotonicity. If $\precsim$ is lower $\tau_d$-continuous, then the metric $d$ satisfies that $d(\mathbf{0},S)>0$ with $d(x,S)=\inf_{y\in S}d(x,y)$, $S=\{x\in X: \sum_{t=1}^{\infty}x_t=1\}$ and $\mathbf{0}=(0,0,\ldots,0,\ldots)$. 
\end{proposition}

From Proposition \ref{Simplex}, Banerjee and Mitra deduced that there is no ethical welfare order satisfying anonymity, strong monotonicity and, in addition, lower $\tau_{d_{c}}$-continuity, $\tau_{d_{s}}$-continuity and $\tau_{d_{p}}$-continuity. Notice that the metrics $d_c,d_s,d_p$ do not hold the simplex condition ``$d(\mathbf{0},S)>0$''.
\medskip

Note that we can restate Proposition \ref{Simplex} interchanging in its statement the lower $\tau_d$-continuity of $\precsim$ by the fact that $\tau$ is finer than $\tau^{\precsim}_{U}$. So if a metric belonging to $\triangle$ violates the simplex condition, then $\tau^{\precsim}_{U}\not\subseteq \tau_d$ necessarily.

\medskip

It must be stressed that our approach  presents two advantages with respect to the approach given by Banerjee and Mitra. On the one hand, it allows us to decide the continuity of the ethical welfare order even if the topology under consideration is not metrizable and the alternative space is $l_{\infty}$ instead of $l_{\infty}^{[0,1]}$. Observe that a few properties that a metric in the class $\triangle$ must satisfied are not true when the intergenerational distributions are not in $l_{\infty}^{+}$. Moreover, every ethical social welfare order $\precsim$ will be continuous with respect to the topology $\tau_d$ induced by a metric (belonging to $\triangle$ or not) on $l_{\infty}$ if and only if $\tau^{\precsim}_{U}\subseteq \tau_d$. On the other hand, contrary to Theorems \ref{Cont1} and \ref{Cont2},  Banerjee and Mitra only provide a necessary condition and  they do not prove the converse of Proposition \ref{Simplex}. Instead, they provide an example of ethical social welfare orders on $l_{\infty}^{[0,1]}$ which is (lower) $\tau_{d_{1}}$-continuous (which satisfies the simplex condition). The aforesaid example is given by the extension of the overtaking type criterion due to Svensson (\cite{svensson}). 

 Svensson proved that every preorder that refines the {\em grading principle} can be extended in such a way that the extension fulfills anonymity and strong monotonicity in \cite{svensson}. The aforementioned grading principle is the preorder $\precsim_{m}$ defined on $l_{\infty}$ as follows:

$$ x\precsim_m y \iff x\leq \pi(y) \, \text{ for some } \pi\in \Pi_{\infty}.$$
However, Example \ref{Eneg} shows that the converse of Proposition \ref{Simplex} does not hold in general.

\begin{example}\label{Eneg} \rm
Let $\precsim^{\frac{1}{2}}$ be the preorder on $l_{\infty}$ defined by 
$$x\precsim^{\frac{1}{2}} y \iff \left\{  \begin{array}{c}
x\precsim_m y, \\
\text{  or  }\\
 \sigma(x)> \sigma(y),
\end{array}\right.,$$ where $\sigma(x)$ denotes the number of coordenates of $x$ which are lower than $\frac{1}{2}$.  Notice that the preorder $\precsim^{\frac{1}{2}}$ is related to the satisfaction of basic needs criterion introduced by G. Chichilnisky in \cite{Chichilnisky1977} (see, also, \cite{Chichilnisky1996}).

Clearly $\precsim^{\frac{1}{2}}$ refines the preorder $\precsim_m$. It is not hard to check that $\precsim^{\frac{1}{2}}$ satisfies anonymity and strong monotonocity (see Proposition \ref{Pminpreorder} below).  By \cite{svensson}, $\precsim^{\frac{1}{2}}$ can be extended in such a way that the extension fulfills anonymity and strong monotonicity (see the paragraph before Proposition \ref{Pminpreorder}). Set $\preceq$ the ethical social welfare order on $l_{\infty}$ that extends $\precsim^{\frac{1}{2}}$.

Now, we define the sequence $(x_n)_{n\in \mathbb{N}}$ in $l_{\infty}$ by 
$$x_n=(\frac{1}{2}-\frac{1}{2^n},\frac{1}{2}-\frac{1}{2^n},0,0,\ldots, 0,\ldots)$$ for each $n\in \mathbb{N}$.  It is clear that  $(x_n)_{n\in \mathbb{N}}$ converges to 
$Z=(\frac{1}{2},\frac{1}{2},0,0,\ldots)$ with respect to the topology $\tau_{d_q}$, since we have that  $d_p(l,y_n)=\frac{2^{\frac{1}{p}}}{2^n}$ for all $n\in \mathbb{N}$.

Set $y=(\frac{1}{2},0,...,0,...)$. It is clear that $x_n\prec^{\frac{1}{2}} y$ and, thus, $x_n\prec y$. Moreover, $y\prec^{\frac{1}{2}} \mathbf{\frac{1}{2}}$ and, thus, $y\prec\mathbf{\frac{1}{2}}$. Whence, $\mathbf{\frac{1}{2}}\in X\setminus L_{\preceq}(y)$ whereas $x_n\notin X\setminus L_{\preceq}(y)$. Therefore $(x_n)_{n\in \mathbb{N}}$ fails to converge with respect to $\tau_U^{\preceq}$. It follows that $\tau_{d_q}$ is not finer than $\tau_U^{\preceq}$. Whence we conclude that the preorder $\preceq$ is not lower $\tau_{d_q}$-continuous.
\end{example}

In the light of the preceding facts, although, as mentioned above, Theorems \ref{Cont1} and \ref{Cont2} characterize the topologies for which an ethical social welfare order is continuous, next we explore the possibility of giving a method, based on Corollary \ref{Alex2}, that warranties the continuity of any extension of an ethical social welfare preorder satisfying anonymity and strong monotonicity on $l_{\infty}$ (not only on $l_{\infty}^{[0,1]}$). The possibility of extending ethical social welfare preorder in such a way the every extension preserves the continuity has attracted the attention of several authors  {(see \cite{Herd,Mash}, among others)}.

%


Next we go one step further than Svensson and we show that every ethical social welfare order on  $l_{\infty}$ satisfying anonymity and strong monotonicity is continuous with respect to every topology finer than the Alexandroff topology induced by the grading principle.

Before stating the announced property, we point out, on account of \cite[Proposition 1]{AsheimBuch01}, that every ethical social welfare order that satisfies anonymity and strong monotonicity refines the grading principle.

\begin{proposition}\label{Pminpreorder}
The relation $\precsim_{m}$ is the smallest ethical social welfare preorder defined on $l_{\infty}$ satisfying anonymity and strong monotonicity, where $\precsim_{m}$ is defined as follows:

$$ x\precsim_m y \iff x\leq \pi(y) \, \text{ for some } \pi\in \Pi_{\infty}.$$
\end{proposition}
\medskip

The following interesting property was proved in \cite[Lemma 1]{BEstRP19} and it will be crucial in order to guarantee the continuity of any extension of the grading principle $\precsim_m$.

\begin{lemma}\label{Lgiltza}
Let $\sqsubseteq$ and $\precsim$ two preorders on a nonempty set $Y$ and $\tau_A^\sqsubseteq$ and $\tau_A^\precsim$ their corresponding Alexandroff topologies. Then the following assertions are equivalent: 
\begin{enumerate}
\item $\sqsubseteq \subseteq \precsim$ ($\precsim$ refines $\sqsubseteq$).

\item  $\tau_A^\precsim\subseteq \tau_A^\sqsubseteq$.
\end{enumerate}
\end{lemma}

It must be pointed out that the upper topology does not fulfill the preceding property such as it has been shown in \cite[Example 4]{BEstRP19}. This fact highlights Corollary \ref{Alex2} against Theorem \ref{Cont2} as the continuity of extensions is under consideration such as happens when ethical social welfare orders which, as Proposition \ref{Pminpreorder} reveals, are extensions of the grading principle.
\medskip

From Corollary \ref{Alex2} and Lemma \ref{Lgiltza} we obtain the promised method that gives the continuity of any ethical social welfare order on $l_{\infty}$.

\begin{proposition}\label{ACont}
Let $\tau$ be a topology on $l_{\infty}$.  If the Alexandroff topology $\tau_A^m$ associated to $\precsim_m$ is contained in $\tau$, then any ethical social welfare order  satisfying anonymity and strong monotonicity is lower $\tau$-continuous.
\end{proposition}

\begin{proof}By Proposition \ref{Pminpreorder} we have that any ethical social welfare order  satisfying anonymity and strong monotonicity refines the grading principle $\precsim_m$. Lemma \ref{Lgiltza} gives that $\tau_A^\precsim\subseteq \tau_A^{\precsim_{m}}\subseteq \tau$. Corollary \ref{Alex2}  provides the lower $\tau$-continuity. \end{proof}

\medskip

In the view of Proposition \ref{ACont}, it is worthy to mention that, although any ethical social welfare order $\precsim$ is lower $\tau$-continuous when $\tau_A^{\precsim{_m}}\subseteq\tau$, in general, there does not exist a lower semicontinuous utility function that represents it. Remember that, according to  \cite{bridges}, for the existence of this utility function, the ethical social welfare order must be perfectly separable.  However, an extension of a preorder that satisfies anonymity and strong monotonicity fails to be separable (in the Debreu sense) in general. Anyway, as exposed in Subsection \ref{NewCharact}, every ethical social welfare order admits a lower semi-continuous multi-utility representation provided that $\tau$ is finer than  $\tau_A^{\precsim{_m}}$ and, thus, finer than $\tau^{\precsim}_{U}$.
\medskip

Since Proposition \ref{Simplex} allows us to discard the topologies induced by the metrics $d_c$, $d_s$ and $d_p$ as an appropriate topology for making lower continuous an ethical social welfare order that fulfills anonymity and strong monotonicity, it seems natural to wonder whether our exposed theory is in accordance with the aforementioned result and, thus, we can infer the same conclusion in our new framework. The next result gives a positive answer to the posed question.

 \begin{proposition}\label{Lsupnorm}Let $\precsim_{m}$ be the grading principle on $l_{\infty}$. The upper topology $\tau_{U}^{\precsim_{m}}$ is not coarser than the topology $\tau_{d_{p}}$. Therefore the  Alexandroff topology $\tau_A^{\precsim{_m}}$ is also not coarser than $\tau_{d_{p}}$.
\end{proposition}

\begin{proof}  Following \cite{svensson}, $\precsim_m$ can be extended in such a way that the extension is a total preorder and fulfills anonymity and strong monotonicity. Set $\preceq$ be such an extension. Thus, $\preceq$ is a ethical social welfare order on $l_{\infty}$. Consider the sequence $(y_n)_{n\in\mathbb{N}}$ in $l_{\infty}$ introduced in Example \ref{SvenEx1} and given as follows:

\begin{center}
 $ y_1=x=(0,1,0,\frac{1}{2}, \frac{2}{2}, 0, \frac{1}{3}, \frac{2}{3}, \frac{3}{3}, 0,...)$,\\
 $ y_2=(\frac{2}{2},1,0,0, \frac{1}{2}, 0, \frac{1}{3}, \frac{2}{3}, \frac{3}{3}, 0,...)$,\\
 $ y_3=(\frac{3}{3},1,0,\frac{1}{2}, \frac{2}{2}, 0, 0, \frac{1}{3}, \frac{2}{3}, 0,...)$,\\
 $...$\\
 $y_n=(\frac{n}{n},1,0, \frac{1}{2}, \frac{2}{2}, 0, \frac{1}{3}, \frac{2}{3}, \frac{3}{3}, 0,...,0,0, \frac{1}{n},..., \frac{n-1}{n},0,...)$
 \end{center}

Set $x=(0,1,0,\frac{1}{2}, \frac{2}{2}, 0, \frac{1}{3}, \frac{2}{3}, \frac{3}{3}, 0,...)$ and $l=(1,1,0,\frac{1}{2}, \frac{2}{2}, 0, \frac{1}{3}, \frac{2}{3}, \frac{3}{3}, 0,...)$. 

It is clear that $y_n\in L^{\precsim_{m}}(x)$ for all $n\in \mathbb{N}$. Thus $y_n\in L^{\preceq}(x)$ for all $n\in \mathbb{N}$. Moreover, the sequence $(y_n)_{n\in \mathbb{N}}$ converges to $l=(1,1,0,\frac{1}{2}, \frac{2}{2}, 0, \frac{1}{3}, \frac{2}{3}, \frac{3}{3}, 0,...)$ with respect to $\tau_{d_{p}}$, since
 $d_p(l,y_n)=\frac{n^{\frac{1}{p}}}{n}$ for all $n\in \mathbb{N}$. 

Nevertheless, the sequence  $(y_n)_{n\in\mathbb{N}}$ fails to converge with respect to $\tau_{U}^{\preceq}$ to $l$, since $y_n\not\in X\setminus L^{\preceq}(x)$ for all $n\in \mathbb{N}$ but $l\in X\setminus L^{\preceq}(x)$ because $x\prec_m l$ and, thus, $x\prec l$. Consequently, 
$\tau_{U}^{\preceq}$ is not coarser than the topology $\tau_{d_{p}}$.  Since $\tau_{U}^{\preceq}\subseteq \tau_A^{\preceq}$ we have that $\tau_A^{\preceq}$ is also not coarser than $\tau_{d_{p}}$ as claimed.  

\end{proof}

\medskip

Proposition \ref{Lsupnorm} explains the reason for which the impossibility results Theorems \ref{Diamond} and \ref{Diamond2} hold.
\medskip

Regarding the possibility of obtaining ethical social welfare orders on $l_{\infty}$ that fulfills anonymity and strong monotonicity and, in addition, they are lower $\tau_{d_{q}}$-continuous or  $\tau_{d_{1}}$-continuous, we have the following. On the one hand, the overtaking type criterion introduced in the proof of Theorem \ref{DiamondNew} is an example of ethical social welfare order on $l_{\infty}$ that satisfies the aforementioned requirements and it is, in addition, lower $\tau_{d_1}$-continuous. It must be stressed that the same fact on $l_{\infty}^{[0,1]}$ was proved in \cite{mitra}. Besides it was shown that $\tau_{d_1}$ is the smallest topology, among the induced by the metrics in $\triangle$, for which exists an ethical social welfare ordering satisfying anonymity, strong monotonicity and being lower continuous. On the other hand, we present an example of an ethical social welfare order on $l_{\infty}$ that satisfies anonymity and strong monotonicity but it is not lower $\tau_{d_{q}}$-continuous. So in the general $l_{\infty}$ framework, $\tau^{\precsim}_U$ is the smallest topology for which there exists an ethical social welfare ordering $\precsim$ satisfying anonymity and strong monotonicity that is lower continuous. Nonetheless, if we restrict ourselves to topologies induced by the metrics in $\triangle$, then again $\tau_{d_{1}}$ is the smallest topology that achieves this end. 

We end the subsection 
recovering Example~\ref{Eneg}, but now we modify it in order to construct an example of ethical social welfare order on $l_{\infty}$ that fails to be lower $\tau_{d_{q}}$-continuous.
To this end, let us introduce a new axiom that we have called {\em negativity}. 
 We shall say that a preorder on a nonempty set $l_{\infty}$ satisfies the negativity if, given ($x,y\in l_{\infty}$), then $x\prec y$ provided that  $\sigma(x)>\sigma(y)$, where $\sigma(x)$ denotes the number of negative coordenates of $x$.
 Then, the preorder $\precsim^+$  on $l_{\infty}$ defined by 
$ x\precsim^+ y $  if and only if $x\precsim_m y$   or  
$ \sigma(x)> \sigma(y),$
%
is actually the smallest preorder satisfying anonymity, strong monotonocity and negativity.  Again, by \cite{svensson}, $\precsim^+$ can be extended in such a way that the extension  $\preceq$  fulfills anonymity and strong monotonicity. 
However, similar as it was done in Example~\ref{Eneg}, it can be proved that the preorder $\preceq$ is not lower $\tau_{d_q}$-continuous.



Hence, we infer that the upper topology is not, in general, coarser than $\tau_{d_q}$. Therefore, it is not possible in general to guarantee the continuity of an ethical social welfare (pre)order on $l_{\infty}$ neither with respecto to $\tau_{d_q}$ nor with respecto to $\tau_{d_{1}}$.
\medskip

Finally, we remark that the negativity axiom could be interpreted from an economical viewpoint as follows: the negative values can be understood as extreme and generalized (that affects to all the generation) cases of war, famine, natural disasters, etc. In case of anonimy data, negative values could suggest losses, debts, bankruptcies, etc.

\section{Quasi-pseudo-metrics: a quantitative tool for reconciling order, topology and preferences}\label{QSMRe}
The manifest difficulty to  reconciling the order and topology when this last is induced by a metric motivates us to leave such structures. The fact, on the one hand, that Theorems \ref{Cont1} and \ref{Cont2} shows that in order to make a preorder continuous is necessary to take into account the upper topology generated by such a preorder and, on the other hand, that the upper topology is not metrizable (notice that it is not Hausdorff) suggests us that the appropriate quantitative tool for reconciling topology and order is exactly provided by quasi-pseudo-metrics, which are able to encode the preorder. 

Following \cite{Ku2} (see also \cite{Goubault-Larrecq13}), a quasi-pseudo-metric on a nonempty set $Y$ is a function $d:Y\times Y\rightarrow \mathbb{R}^{+}$ such that for
all $x,y,z\in Y: $
\medskip

$
\begin{array}{ll}
\text{\textrm{(i)}} & d(x,x)=0, \\
\text{\textrm{(ii)}} & d(x,z)\leq d(x,y)+d(y,z).
\end{array}
 $
\medskip

Each quasi-pseudo-metric $d$ on a set $X$ induces a topology $\tau_d $ on $Y$ which has as a base the family of open balls $\{B_{d}(x,\varepsilon): x\in X \text{ and } \varepsilon>0\}$, where $B_{d}(x,\varepsilon)=\{y\in X:d(x,y)<\varepsilon\}$ for all $x\in X$ and $\varepsilon>0$. 

A quasi-pseudo-metric space is a pair $(Y,d)$ such that $Y$ is a nonempty set and $d$ is a quasi-metric on $Y$.

Notice that the topology $\tau_{d}$ is $T_0$ if and only if $d(x,y)=d(y,x)=0$ for all $x,y\in Y$.

Observe that a pseudo-metric $d$ on a nonempty set $Y$ is a quasi-pseudo-metric which enjoys additionally the following properties for all $x,y\in Y$:
\medskip

$
\begin{array}{ll}
\textrm{(iii)} & d(x,x)=0 \Rightarrow x=y,\\
\textrm{(iv)} & d(x,y)=d(y,x). 
\end{array} $
\medskip

A metric is a pseudo-metric $d$ on a nonempty set $Y$ which, in addition, fulfills for all $x,y\in Y$ the property below:
\medskip

$
\begin{array}{ll}
\textrm{(v)} & d(x,y)=0\Rightarrow x=y,\\
\end{array} $
\medskip

If $d$ is a quasi-pseudo-metric on a set $Y$, then the function $d^{s}$ defined on $Y\times Y$ by $d^{s}(x,y)=\max\{d(y,x),d(x,y)\}$ for all $x,y\in Y$ is a pseudo-metric on $Y$.

Every quasi-pseudo-metric space $d$ on $Y$ induces a preorder $\precsim_{d}$ which is defined on $Y$ as follows: $x\precsim_{d} y \Leftrightarrow d(x,y)=0$.

An illustrative example of quasi-pseudo-metric spaces is given by the pair $(\mathbb{R},d_L)$, where $d_L(x,y)=\max\{x-y,0\}$ for all $x,y\in \mathbb{R}$. Observe that $\tau_{d_{L}}$ is the upper topology $\tau_{U}^{\leq}$ on $\mathbb{R}$, where $\leq$ stands for the usual preorder on $\mathbb{R}$. Note that $d^{s}_L(x,y)=|y-x|$ for all $x,y\in \mathbb{R}$.

Following \cite{BonBreuRutt1996}, every preorder $\precsim$ can be encoded by means of a quasi-pseudo-metric. Indeed, if  $\precsim$ is a preorder on $X$, then the function $d_\precsim\colon X\times X \to \mathbb{R}^{+}$ given by

\begin{center}
$d_\precsim(x, y) = \left\{  \begin{array}{ll}
0, & x \precsim y\\ 
1, &  otherwise 
\end{array}\right.$\\
\end{center} is a quasi-pseudo-metric on $X$. 

Obviously, $x\precsim_{d_{\precsim}}y \Leftrightarrow  d_{\precsim}(x, y)=0 \Leftrightarrow x\precsim y$ and, in addition, we have that $\tau_{d_{\precsim}}=\tau_A^{\precsim}$ and that $\tau_U^{\precsim},\tau_L^{\precsim}\subseteq \tau_{d^s_{\precsim}}$. Therefore, Corollaries \ref{Alex1} and \ref{Alex2} give respectively the $\tau_{d^s_{\precsim}}$-continuity and the lower $\tau_{d_{\precsim}}$-continuity of $\precsim$. 

It must be stressed that (pseudo-)metrics are not able to encode any preorder except the equality order $\precsim_{=}$, that is, $x\precsim_{=}y \Leftrightarrow x=y$.

In view of the exposed facts, the use of quasi-pseudo-metrics makes possible to reconcile ``metric methods'' of measure and order. In the particular case of intergenerational distribution problem, this generalized metrics help us to provide both things, the numerical quantifications about the increase of welfare and the arrow of such an increase. Note that a metric would be able to give information on the increase but it, however, it will not give the aforementioned arrow.

The preceding method of ``metrization'' is able to guarantee, in contrast to Theorems \ref{Cont1} and \ref{Cont2}, possibility counterparts of the celebrate impossibility theorems due to Diamond, Svensson and Sakai introduced in Subsection \ref{Possi} in a appropriate metric approach. Specifically we obtain combining the preceding quasi-pseudo-metrization and Corollaries \ref{Alex1} and \ref{Alex2}, the next result which translates Theorems \ref{DiamondNew}, \ref{SakaiNew} and \ref{SakaiNew2} into the quantitative framework.

\begin{theorem}\label{possibilityQPM}There exists an ethical social welfare order $\precsim$ on $l_{\infty}$ which satisfies anonymity, strong monotonicity, strong distributive fairness semi convexity and $\tau_{d^s_{\precsim}}$-continuity and, thus, lower $\tau_{d_{\precsim}}$-continuity. 
\end{theorem}
\begin{proof}It is enough to observe that $\tau_U^{\precsim},\tau_L^{\precsim}\subseteq \tau_{d^s_{\precsim}}$  and $\tau_U^{\precsim}\subseteq \tau_{d_{\precsim}}$.
\end{proof}

\medskip

Returning to the discussion made in Subsection \ref{Possi} about the continuity of any extension of an ethical social welfare preorder satisfying anonymity and strong monotonicity on $l_{\infty}$ we have the following.

\begin{theorem}Let $\precsim_m$ be the smallest preorder on $X$ satisfying anonymity and strong monotonicity on $l_{\infty}$. Then any other ethical social welfare (pre)order satisfying anonymity and strong monotonicity is $\tau_{d^s_{\precsim_m}}$-continuous on $l_{\infty}$ and, thus, lower $\tau_{d_{\precsim_m}}$-continuous.
\end{theorem}

\begin{proof}The desired result follows from Corollary \ref{Alex1} and Lemma \ref{Lgiltza}.

\end{proof}

\medskip

In the light of the above facts and the fact that a preorder is lower $\tau$-continuous with respect to a topology only in the case such topology refines the upper topology induced by the the preorder, it seems natural to restrict attention to the use of quasi-pseudo-metrics as a quantitative tool that allows us, at the same time, to get a numerical quantification of the improvement of welfare and of the closeness between intergenerational distributions.

\section{Order, topology and preferences: going back to metrics}\label{s5}
In Section \ref{QSMRe} we have shown that the use of quasi-pseudo-metrics reconciles ``metric methods'' of measuring and order requirements of ethical social welfare preorders. With the aim of  keeping close to the classical way of measuring in the literature, that is through metrics, a refinement of the method that encode the preorder in such a way that classical metrics are involved. The economical interpretations of their quantifications are also exposed.

 \medskip

In the remainder of this section we introduce a collection of techniques which generate quasi-pseudo-metrics from a given preorder and a metric on a nonempty set. The aforesaid quasi-pseudo-metrics generate either the Alexandroff topology induced by the preorder or a topology finer than it. So the below techniques provide the lower continuity of the preorder. 
\medskip

In order to state the mentioned techniques let us recall that, following \cite{Goubault-Larrecq13}, a quasi-metric on a nonempty set $Y$ is a quasi-pseudo-metric on $Y$ such that, for all $x,y\in Y$, the following property is hold:
\medskip

$\text{\textrm{(vi)} \hspace{0.2cm}} d(x,y)=d(y,x)=0\Leftrightarrow x=y.$ 
\medskip

A quasi-metric is called $T_1$ provided, for all $x,y\in Y$, that next property is true:
\medskip

$\text{\textrm{(vii)} \hspace{0.2cm}} d(x,y)=0\Leftrightarrow x=y.$ 
\medskip

Notice that the topology $\tau_d$ is $T_0$ when the quasi-pseudo-metric is just a quasi-metric and, in addition, such a topology is $T_1$ when the quasi-metric is $T_1$.
\medskip

Taking this into account we have the next result. Before stating it, let us recall that a pseudo-metric space $(Y,d)$ is $1$-bounded whenever $d(x,y)\leq 1$ for all $x,y\in Y$.

\begin{theorem}\label{L1}
Let $(Y,d)$ be a $1$-bounded pseudo-metric space and let $\precsim$ be a preorder on $Y$.  Then, the function $d^1_\precsim\colon X\times X \to \mathbb{R}^{+}$ defined by   \\
\begin{center}

$d^1_\precsim(x, y) = \left\{  \begin{array}{lc}
 d(x,y), & x \precsim y\\
 \\ 1,  &  otherwise
\end{array}\right..\medskip$\\
\end{center} is a quasi-pseudo-metric such that $\tau_{d^{1}_{\precsim}}$ is finer than $\tau_{A}^\precsim$. Therefore, $\precsim$ is lower $\tau_{d_\precsim}$- continuous. If $d$ is a metric on $Y$, then $d^1_\precsim$ is a $T_1$ quasi-metric.
\end{theorem}

\begin{proof}
The function $d^1_\precsim$ is actually a quasi-pseudo-metric. To see that, notice that in the case $x\precsim y\precsim z$ the triangular inequality  $d^1_\precsim(x, z)\leq d^1_\precsim(x, y)+ d^1_\precsim(y,z)$ is satisfied due the fact that $d$ is a metric. In any other case, etiher $d^1_\precsim(x, y)=1$ or $d^1_\precsim(y,z)=1$ and, hence, triankle inequality is satisfied too. Since $d^1_\precsim(x, x)=0 \Leftrightarrow d(x,x)=0$ for any $x\in Y$, we conclude that it is actually a quasi-pseudo-metric.
\medskip

Of course if $d$ is a metric on $Y$, then $d^1_\precsim(x, y)=0 \Leftrightarrow d(x,y)=0$ for any $x,y\in Y$. It follows that $d^1_\precsim$ is actually a $T_1$ quasi-metric.
\medskip

Let's see now that $\tau_{d^1_\precsim}$ is finer than $\tau_{A}^\precsim$. To this end, let $O\in \tau_{A}^\precsim$ and $x\in O$. Then $O=\bigcup_{x\in O} U^{\precsim}(x)$. Fix $r<1$. Then $B_{d^1_\precsim}(x,r)\subseteq U^\precsim(x)\subseteq O$. Hence, we conclude that  $\tau_{A}^\precsim\subseteq \tau_{d^1_\precsim}$. By Corollary \ref{Alex2} we have the lower $\tau_{d^1_\precsim}$- continuity.
\end{proof}
\medskip

Regarding intergenerational distributions, the quasi-pseudo-metric $d^1_\precsim$ introduced in the previous result is able to quantify the increase of welfare (when $x\precsim y$), by means the use of a metric. Moreover, it differentiates this case from the rest of the cases assigning, the retrogress ($y\prec x$) and the incomparability ($x\bowtie y$), to all of them $1$ as a quantification. 
\medskip

A slight modification of the technique introduced in Theorem \ref{L1} gives the next one.

\begin{theorem}\label{L2}
Let $(Y,d)$ be a pseudo-metric space and let $\precsim$ be a preorder on $Y$.  Then, the function $d^2_\precsim\colon X\times X \to \mathbb{R}^{+}$ defined by   \\

\begin{center}
$d^2_\precsim(x, y) =  \left\{  \begin{array}{lc}
\frac{d(x,y)}{2}, & x \precsim y\\ 
  \frac{1}{2} + \frac{d(x,y)}{2},&   otherwise
\end{array}\right..$
\end{center} is a quasi-pseudo-metric such that $\tau_{d^2_\precsim}$ is finer than $\tau_{A}^\precsim$. Therefore, $\precsim$ is lower $\tau_{d^2_\precsim}$- continuous. If $d$ is a metric on $Y$, then $d^1_\precsim$ is a $T_1$ quasi-metric.
\end{theorem}

\begin{proof}Next we show that $d^2_\precsim$ is a quasi-metric. Indeed  the triangular inequality $d^2_\precsim(x, z)\leq d^2_\precsim(x, y)+ d^2_\precsim(y,z)$  holds whenever $x\precsim y\precsim z$, since $d$ is a metric. In any other case, etiher $d^2_\precsim(x, y)=\frac{1}{2}+\frac{d(x,y)}{2}$ or $d^2_\precsim(y,z)=\frac{1}{2}+\frac{d(y,z)}{2}$ and, hence, $\frac{1}{2}+\frac{d(x,y)}{2}+\frac{d(y,z)}{2}\leq d^2_\precsim(x, y)+ d^2_\precsim(y,z)$. Since $d(x,y)\leq d(x,y)+d(y,z)$ we deduce that $d^2_\precsim(x,z)\leq \frac{1}{2}+\frac{d(x,y)}{2}+\frac{d(y,z)}{2}$ and, hence, $d^2_\precsim(x, z)\leq d^2_\precsim(x, y)+ d^2_\precsim(y,z)$. 

The same arguments to those given in the proof of Theorem \ref{L1} apply in order to show that $d^2_\precsim(x,x)=0\Leftrightarrow d(x,x)=0$ for all $x\in Y$ and that  $d^2_\precsim(x,y)=0\Leftrightarrow x=y$ whenever $d$ is a metric on $Y$ and, in addition, that $\tau_{d^2_\precsim}$ is finer than $\tau_{A}^\precsim$. Therefore, $\precsim$ is lower $\tau_{d^2_\precsim}$- continuous.

\end{proof}

\medskip

In the same way that $d^1_\precsim$, when intergenerational distributions are under consideration, the quasi-pseudo-metric $d^2_\precsim$ is able to quantify, by means of a metric, the increase of welfare (when $x\precsim y$). Moreover, it differentiates this case from the rest of the cases, the retrogress ($y\prec x$) and the incomparability ($x\bowtie y$). However, this time it assigns a lower value for the former case.

\medskip

Notice that, among the possible metrics, those belonging to the Banerjee and Mitra class $\triangle$ can be considered in statement of Theorem \ref{L1} and \ref{L2}.
\medskip

It must be stressed that modifications of the preceding technique can be obtained proceeding as follows:

\begin{center}
$d^2_\precsim(x, y) =  \left\{  \begin{array}{lc}
\frac{k\cdot d(x,y)}{n}, & x \precsim y\\ 
\frac{k}{n} + \frac{(n-k)\cdot d(x,y)}{n}, &   otherwise 
\end{array}\right..\medskip$\\
\end{center} for some $n\in \mathbb{R}_+$ and $k\in [0,n]$. 

\medskip

The next result introduces a technique which is related to the methods exposed in \cite{Levin1984,Levin1991,Levin1997}.

\begin{theorem}\label{Lrefinado}
Let $\precsim$ be a preorder on $Y$. 
If $u\colon (X, \leq)\to (0,1)$ is a weak-utility for $\precsim$, then the function $d^3_\precsim\colon X\times X \to \mathbb{R}^{+}$ defined by   \\

 \begin{center}
 
$d^3_\precsim(x, y) =  \left\{  \begin{array}{lc}
 0, & x \precsim y 
 \\
1+|u(x)-u(y)|, & y\prec x
 \\ 1, &  otherwise
\end{array}\right..\medskip$
 \end{center} is a quasi-pseudo-metric such that $\tau_{d^3_{\precsim}}=\tau_A^\precsim$. Therefore, $\precsim$ is lower $\tau_{d^3_\precsim}$- continuous.
\end{theorem}

\begin{proof} It is trivial that $d^3_\precsim(x,x)=0$, for any $x\in Y$.
Let's see that the triangular inequality is satisfied, i.e., that $$d^3_\precsim(x,z)\leq d^3_\precsim(x,y)+d^3_\precsim(y,z)$$ for any $x,y,z\in X$. For this propose, we set $d(x,y)=|u(x)-u(y)|$ for all $x,y\in Y$ and distinguish the following possible cases.
\medskip

\noindent Case 1. $x\precsim z$. Then the inequality is trivially satisfied.
\medskip

\noindent Case 2. $x\bowtie z$. Then $d^3_\precsim(x,z)=1$. Notice that the case $x\precsim y\precsim z$ is impossible. Then the following cases may hold:

\begin{enumerate}
\item[$(i)$] If  $x\bowtie y$ or  $y\bowtie z$, then the inequality is satisfied because we have either $d^3_\precsim(x,y)=1$ or $d^3_\precsim(y,z)=1$.

\item[$(ii)$] If $x\precsim y$, then we have that $\neg(y\precsim z)$. In fact, we have that $z\prec y$, otherwise we would be either in case $(i)$ or in the impossible case $x\precsim y\precsim z$. Therefore, we obtain that $1\leq 1+d(y,z)$ and, thus, the inequality is satisfied.
\item[$(iii)$] If $\neg(x\precsim y)$, then we have that $y\prec x$, otherwise we would be in case $(i)$ above. Hence, we have that either $y\precsim z$ or $z\prec y$. Observe that $y\bowtie z$ matches up with the case $(i)$. Thus if $y\precsim z$ then we obtain $d^3_\precsim(x,y)=1+d(x,y)$, $d^3_\precsim(y,z)=0$ and, therefore, the inequality holds becase $1\leq 1+d(x,y)$. Finally, if $z\prec y$ then we obtain $z\prec y \prec x$ which contradicts the hypothesis $x\bowtie z$.

\end{enumerate}
\medskip

\noindent Case 3. $z\prec x$. Then $d^3_\precsim(x, z) =1+ d(x,z)$ and the following cases may hold:

\begin{enumerate}
\item[$(i)$] If  $x\bowtie y$ as well as  $y\bowtie z$, then $d^3_\precsim(x, y)=d^3_\precsim(y, z)=1$ and, thus, the inequality is satisfied because $1+d(x,y)\leq 2$. 

\item[$(ii)$]  If  $x\bowtie y$ or  $y\bowtie z$, then we have the following cases:

\begin{enumerate}
\item[$(ii_1)$] If  $z\bowtie y$, then  $y\prec  x$. In this case the inequality is satisfied because $d^3_\precsim(x, y)=1+d(x,y)$, $d^3_\precsim(y,z)=1$ and, thus, $1+d(x,z)\leq 2+d(x,y)$. 

\item[$(ii_2)$]  If  $x\bowtie y$,  then  $z\prec  y$. In this case the inequality is satisfied too, since $d^3_\precsim(x, y)=1$, $d^3_\precsim(y,z)=1+d(y,z)$ and, thus, $1+d(x,z)\leq 2+d(y,z)$.

\end{enumerate}

\item[$(iii)$] If it holds neither  $x\bowtie y$ nor  $y\bowtie z$, then we have the following cases:
\begin{enumerate}
\item[$(iii_1)$] If  $z\prec y\prec x$, then $1+d(x,z)=d^3_\precsim(x,z)\leq 1+d(x,y)+d(y,z)\leq d^3_\precsim(x,y)+d^3_\precsim(y,z)$.

\item[$(iii_2)$]  If  $y\precsim z\prec x$, then $d^3_\precsim(x,z)=1+d(x,z)\leq 1+d(x,y)=d^3_\precsim(x,y)+d^3_\precsim(y,z)$ with $d^3_\precsim(y,z)=0$.

\item[$(iii_3)$]  If  $z\prec x\precsim y$, then $d^3_\precsim(x,z)=1+d(x,z)\leq 1+d(y,z)=d^3_\precsim(x,y)+d^3_\precsim(y,z)$ with $d^3_\precsim(x,y)=0$.
\end{enumerate}

\end{enumerate}

Therefore, taking into account all above studied cases, we conclude that $d^3_\precsim$ satisfies the triangular inequality and, hence, it is actually a quasi-pseudo-metric.
\medskip

Finally, it remains to prove that $\tau_{d^3_\precsim}\subseteq \tau_{A}^\precsim$. The fact that $\tau_{A}^\precsim\subseteq \tau_{d^3_\precsim}$ can be deuced following the same arguments applied to the proof of Theorem \ref{L1}. Next we show that $\tau_{d^3_\precsim}\subseteq \tau_{A}^\precsim$. Thus, consider $A\in \tau_{d^3_\precsim}$. Then, for each $x\in A$, there exists $0<\varepsilon<1$ such that $B_{d^3_\precsim}(x,\varepsilon)\subseteq A$. Clearly, $U^{\precsim}(x)\subseteq B_{d^3_\precsim}(x,\varepsilon)\subseteq A$. So  $A\in \tau_{A}^\precsim$. Whence we conclude that $\tau_{d^3_\precsim}\subseteq \tau_{A}^\precsim$.

\end{proof}

\medskip

Similar to $d^1_{\precsim}$ and $d^2_{\precsim}$, the quasi-pseudo-metric $d^3_{\precsim}$ quantifies, by means of a metric, the increase of welfare $x\precsim y$ when  intergenerational distributions are under consideration. Moreover, it differentiates this case from the rest of the cases, the retrogress ($y\prec x$) and the incomparability ($x\bowtie y$). But now it assigns a greater and constant value $1$ when we want to measure the distance between incomparable elements and even a bigger value in case of regression.
\medskip

The quasi-pseudo-metric $d^2_\precsim$ introduced in Theorem \ref{L2} can be modified in such a way that its quantifications can be understood in the spirit of the quasi-pseudo-metric $d^3_\precsim$ of Theorem \ref{Lrefinado} such as the next result shows.

\begin{theorem}\label{FAP5}
Let $\precsim$ be a preorder on $Y$.
If $u\colon (X, \leq)\to (0,1)$ is a weak-utility for $\precsim$,  then the function $d^4_\precsim\colon X\times X \to \mathbb{R}^{+}$ defined by   \\

 \begin{center}
$d^4_\precsim(x, y) =  \left\{  \begin{array}{lc}
\frac{u(y)-u(x)}{2},& x \precsim y,
 \\
\frac{1}{2}+  \frac{u(x)-u(y)}{2},& y\prec x,
 \\ \frac{1}{2},&  otherwise,
\end{array}\right. \medskip$

 \end{center}  is a quasi-pseudo-metric such that $\tau_{d^4_{\precsim}}$ is finer than $\tau_A^\precsim$. Therefore, $\precsim$ is lower $\tau_{d^4_\precsim}$- continuous.
\end{theorem}

\begin{proof}
The proof is similar to the proof of Theorem \ref{Lrefinado}.
\end{proof}

\medskip

Finally we obtain the following interesting onsequence.

\begin{corollary}Any ethical social welfare preorder satisfying anonymity and strong monotonicity is lower $\tau_{d^i_{\precsim_m}}$-continuous with $i=1,2,3,4$.
\end{corollary}

\begin{proof}By Theorems \ref{L1}, \ref{L2}, \ref{Lrefinado}, $d^i_{\precsim_{m}}$ is a quasi-metric whose topology $\tau_{d^{i}_{\precsim_{m}}}$ is finer than or equal to $\tau_{A}^{\precsim_{m}}$ for all $i=1,2,3,4$. 

By Proposition \ref{Lsupnorm} we have that every ethical social welfare preorder $\preceq$ satisfying anonymity and strong monotonicity is an extension of  $\precsim_{m}$. Thus, by  Lemma \ref{Lgiltza}, we obtain that $\tau_A^{\preceq}\subseteq \tau_A^{\precsim_{m}}\subseteq  \tau_{d^i_{\precsim_m}}$ for all $i=1,2,3,4$. This concludes the proof.

\end{proof}

\section{Conclusion}

Summarizing, in the present paper we have studied the compatibility between preorders and  topologies. Thus, we have provided a characterization of those that are continuous-complatible. Such a characterization states that the considered  topologies must be finer than the so-called upper topology induced by the preorder and, thus, it  clarifies which topology is the smallest one among those that make a preorder. Moreover, we have given sufficient conditions that allows us to discard in an easy way the continuity of a preorder. Of course, such a characterization is applied to provide an explanation about the reason for which it is not possible (in general) to merge a social intergenerational preference which satisfies Pareto efficiency and anonymity with the continuity axiom. Thus, possibility counterparts of the impossibility theorems due to Diamond, Svensson and Sakai are provided. Besides, we have shown that our methodology is in accordance with the classification due to Banerjee and Mitra. However, we have shown that our characterization presents two advantages with respect to the approach given by the aforesaid authors. On the one hand, the new result allows us to decide the continuity of the preference even if the topology under consideration is not metrizable. On the other hand, Banerjee and Mitra only provide a necessary condition. In this direction we have provided an example of social intergenerational preference that enjoy anonymity and strong Pareto requirements and, in addition, it fulfills the simplex condition, due to Banerjee and Mitra, but they are not continuous.

As a matter of the above exposed facts and the fact that the upper topology is not metrizable, we have suggested quasi-pseudo-metrics as an appropriate quantitative tool for reconciling topology and social intergenerational preferences. Concretely, we have shown that such generalized metric notion is able to encode the order relation that induce the intergenerational preference. Thus it provides numerical quantifications about the increase of welfare and the arrow of such an increase. Note that a metric would be able to yield information on the increase but it, however, will not give the aforementioned arrow.

Based on the fact that every preorder, and thus every social intergenerational preference, can be encoded by means of a quasi-pseudo-metric, we have developed a method to induce a quasi-pseudo-metric that makes always the preference continuous with respect its induced topology,  the Alexandroff topology generated by the preorder, which is finer than the upper topology. Such a method is able to guarantee possibility counterparts of the celebrate impossibility theorems due to Diamond, Svensson and Sakai and, in addition, it is able to give numerical quantifications of the improvement of welfare. Moreover, we have also shown that our method makes always the preferences semi-continuous multi-utility representables in the sense of \"{O}zg\"{u} Evern and Efe O. Ok.

Finally, a refinement of the previous method is also presented in such a way that metrics are involved. 

 \section*{Acknowledgements} 

Asier Estevan acknowledges financial support from the Ministry of Economy and Competitiveness of Spain under grants   MTM2015-63608-P (MINECO/FEDER) and ECO2015-65031.

Oscar Valero acknowledges financial support from  FEDER/Mi\-nisterio de Ciencia, Innovaci\'{o}n y Universi\-dades-Agencia Estatal de Investigaci\'{o}n/$_{-}$\-Pro\-yecto PGC2018-095709-B-C21 and from project PROCOE/4/2017 (Direcci\'{o} General d'Innovaci\'{o} i Recerca, Govern de les Illes Balears).

\bibliographystyle{IEEEtran}

\begin{thebibliography}{1}

 \bibitem{AsheimBuch01}G.B. Asheim, W. Buchholf, {\em Justifying Sustainabiklity}, Journal od Envieronmental Economics and Management \textbf{41} (2001), 252-268.
 
\bibitem{mitra} K.  Banerjee,  T. Mitra, \emph{On the continuity of ethical social welfare orders on infinite utility streams}, Social Choice Welfare \textbf{30}, (2008) 1-12.
 
\bibitem{mitra2} K. Basu, T. Mitra,  \emph{Aggregating infinigte utility streams with intergenerational equity: the impossibility of being paretian}, Econometrica \textbf{71} (5) (2003), 1557-1563.
 
 \bibitem{BonBreuRutt1996}M.M. Bonsangue, F. van Breugel, J.J.M.M. Rutten, {\em Alexandroff and Scott topologies for generalized ultrametric spaces}, Annals of the New York Academy of Sciences  \textbf{806} (1) (1996), 49-68.

 \bibitem{BEstRP19}G. Bosi, A. Estevan, A. Ravent\'{o}s-Pujol,  {\em Topologies for semicontinuous Ritcher-Peleg multi-utilities}, Theory and Decision  \textbf{88} (2020) , 457-470.
 
\bibitem{bridges} D.S. Bridges and G.B. Mehta, {\em Representations of Preference Orderings},  Berlin: Springer-Verlag, 1995. 
 
 \bibitem{Campbell1985}D.E. Campbell, {\em Impossibility theorems and infinite horizon planning},  Social Choice and Welfare \textbf{2} (1985), 283-293.

\bibitem{Chichilnisky1977}G. Chichilnisky, {\em Economic development and efficiency criteria in the satisfaction of basic needs},  Applied Mathematical Modelling \textbf{1} (6)(1977), 290-297.

\bibitem{Chichilnisky1996}G. Chichilnisky, {\em An axiomatic approach to sustainable development}, Social Choice and Welfare \textbf{13} (1996), 231-257.

\bibitem{Diamond1965}P. Diamond, {\em The evaluation on infinite utility streams}, Econometrica, \textbf{33} (1965), 170-177.

\bibitem{Epstein1986}L.G. Epstein, {\em Intergenerational preference orderings},  Social Choice and Welfare \textbf{3} (1986), 151-160.
 
 \bibitem{EvO11} O.  Evren,   E.A. Ok, {\em On the multi-utility representation of preference relations}, Journal of Mathematical Economics {\bf 47} (2011),  554-563.


\bibitem{FaugerasRuschendorf18}O.P. Faugeras, L. R\"{u}schendorff, {\em Risk excess measures induced by hemi-metrics}, Probability, Uncertainty and Quantitative Risk \rm (2018), 3:6.


\bibitem{FM03}M. Fleurbaey, Ph. Michel, {\em Intertemporal equity and the extension of the Ramsey criterion}, Journal of Mathematical Economics \textbf{39} \rm (2003), 777-802.

\bibitem{Gierzetall2003}G. Gierz, K.H. Hofmann, K. Keimel, J.D. Lawson, M.W. Mislove, D.S. Scott, {\em Continuous Lattices and Domains}, Cambridge: Cambridge University Press,  2003. 

\bibitem{Goubault-Larrecq13}J. Goubault-Larrecq, {\em Non-Hausdorff Topology and Domain Theory}, New York: Cambridge University Press, 2013


\bibitem{Herd} G. Herden,  {\em On the existence of utility functions II}, Math. Soc. Sci. \textbf{18} (1989), 107-117.


\bibitem{Koopmans1960} T. Koopmans,  {\em Stationary ordinal utility and impatience}, Econometrica  \textbf{28} \rm (1960), 287-309.


\bibitem{Ku2} H.P.A. Künzi, {\em Nonsymmetric distances and their associated topologies: About the origins
of basic ideas in the area of asymmetric topology}, Handbook of the History of General
Topology, Aull, C.E. and Lowen R. (Eds.), vol.
3, (Kluwer, Dordrecht, 2001), 853-968.


\bibitem{Lawers1997} L. Lauwers {\em Continuity and equity with infinite horizons}, Social Choice Welfare  \textbf{14} (1997), 345-356.

\bibitem{Levin1984} V. L. Levin, \emph{Lipschitz pre-orders and Lipschitz utility functions}, Russian Mathematical Surveys \textbf{39} \rm (1984), 217-218.

\bibitem{Levin1991} V. L. Levin, \emph{Some applications of set-valued mappings in mathematical economics}, Journal of Mathematical Economics \textbf{20} \rm (1991), 69-87.

\bibitem{Levin1997} V. L. Levin, \emph{Reduced cost functions and their applications}, Journal of Mathematical Economics \textbf{28} \rm (1997), 155-186.


\bibitem{Levin2008} V. L. Levin, \emph{Smooth feasible solutions to a dual Monge-Kantorovich problem and their application to the best approximation and mathematical economics problems}, Doklady Mathematics \textbf{77}(2) \rm (2008), 281-283.

\bibitem{Levin2011} V. L. Levin, \emph{General preferences and utility functions . An approach based on the dual Kantarovich problem}, Doklady Mathematics \textbf{83}(2) \rm (2011), 236-237.


\bibitem{Levin} V. L. Levin, \emph{Functionally closed preorders and strong stochastic dominance}, Soviet Math. Doklady \textbf{32} \rm (1985), 22-26.


\bibitem{Levin2} V. L. Levin, \emph{The Monge-Kantorovich problems and stochastic preference relation}, Advances in
Mathematical Economics, \textbf{3} \rm (2001), 97-124.

\bibitem{Mas-ColellWhinstonGreen1995}A. Mas-Colell, M.D. Whinston, J.R. Green, {\em Microeconomic Theory}, Oxford: Oxford University Press, 1995.




\bibitem{Mash}  Mashburn, J.D.,  A note on reordering ordered topological spaces and the
existence of continuous, strictly increasing functions, {\em Topology Proceedings} {\bf 20} (1995), 207-250.

\bibitem{GMetha1998}G.B. Metha, {\em Preference and utility}, Handbook of Utility Theory, S. Barbet\`{a} et all. (Eds.), Vol. 1, Kluwer Academic Publishers, Dordrecht, 1998, pp. 1-50.

\bibitem{Ok07}E.O. Ok. {\em Real Analysis with Economics Applications}, Princeton: Princeton University Press, 2007.

\bibitem{Peleg1970}B. Peleg, {\em Utility functions for partially ordered topological spaces}, Econometrica \textbf{38} (1970), 93-96.


\bibitem{StoyanovRachevFabozzi11}S.T. Rachev, S.V. Soyanov, F.J. Fabozzi, {\em A Probability Metrics Approach to Financial Risk Measures}, Oxford: Wiley-Blackwell, 2011. 

\bibitem{Ritcher1966}M. Ritcher, {\em Revealed preference theory}, Econometrica \textbf{34} (1966), 635-645.

\bibitem{sakaiSC} T. Sakai, {\em An axiomatic approach to intergenerational equity}, Social Choice Welfare,  \textbf{20} (2003), 167-176

\bibitem{sakaiEB} T. Sakai,  (2003) {\em{Intergenerational preferences and sensitivity to the present}}, Economics Bulletin, \textbf{4} (26) (2003), 1-5. 

\bibitem{Sakai06} T. Sakai, {\em Equitable intergenerational preferences on restricted domains}, Social Choice Welfare  \textbf{27} (2006), 41-54.

\bibitem{Sidgwick1970}S. Sidgwick, {\em The Methods of Ethics}, Macmillan, London, 1999.

\bibitem{StoyanovRachevFabozzi12}S.V. Soyanov, S.T. Rachev, F.J. Fabozzi, {\em Metrization of stochastic dominance rules}, International Journal of Theoretical and Applied Fianance \textbf{15}(2), 1250017-1.

\bibitem{svensson} L.G.  Svensson,(1980) {\em{Equity Among Generations}}, Econometrica \textbf{48}, 1251-1256.

\bibitem{Varian1992}H.R. Varian, {\em Microeconomics Analysis}, London: W.W. Norton \& Company, 1992.
\end{thebibliography}

\end{document}